\documentclass[12pt]{article}
\usepackage[utf8]{inputenc}
\usepackage[dvipsnames]{xcolor}
\usepackage{amsmath,amsthm,amssymb,mathrsfs,url,graphicx,hyperref}

\title{GAP Measures and Wave Function Collapse}

\author{
Roderich Tumulka\footnote{Fachbereich Mathematik, Eberhard Karls University T\"ubingen, Auf der Morgenstelle 10, 72076 T\"ubingen, Germany. E-mail: roderich.tumulka@uni-tuebingen.de}
}

\date{November 2, 2026}

\theoremstyle{plain}
\newtheorem{thm}{Theorem}

\theoremstyle{definition}
\newtheorem{rem}{Remark}

\newcommand{\Hilbert}{\mathscr{H}}

\newcommand{\be}{\begin{equation}}
\newcommand{\ee}{\end{equation}}

\DeclareMathOperator{\tr}{tr}
\newcommand{\EEE}{\mathbb{E}}
\newcommand{\PPP}{\mathbb{P}}
\newcommand{\RRR}{\mathbb{R}}
\newcommand{\SSS}{\mathbb{S}}
\newcommand{\CCC}{\mathbb{C}}

\newcommand{\vx}{\boldsymbol{x}}

\newcommand{\vq}{{\boldsymbol{q}}}

\newcommand{\GAP}{\mathrm{GAP}}
\newcommand{\GA}{\mathrm{GA}}
\newcommand{\G}{\mathrm{G}}
\newcommand{\sF}{\mathscr{F}}
\newcommand{\sX}{\mathscr{X}}

\begin{document}
\maketitle
\begin{abstract}
GAP measures (also known as Scrooge measures) are a natural class of probability distributions on the unit sphere of a Hilbert space that come up in quantum statistical mechanics; for each density matrix $\rho$ there is a unique measure $\GAP_\rho$. We describe and prove a property of these measures that was not recognized so far: If a wave function $\Psi$ is $\GAP_\rho$ distributed and a collapse occurs, then the collapsed wave function $\Psi'$ is again GAP distributed (relative to the appropriate $\rho'$). This fact applies to collapses due to a quantum measurement carried out by an observer, as well as to spontaneous collapse theories such as CSL or GRW. More precisely, it is the conditional distribution of $\Psi'$, given the measurement outcome (respectively, the noise in CSL or the collapse history in GRW), that is $\GAP_{\rho'}$. 

\medskip

\end{abstract}

\section{Introduction}
\label{sec:intro}

We report a connection between quantum statistical mechanics and the foundations of quantum mechanics that was surprising to us. It is a connection between GAP measures, certain probability distributions that play a role in quantum statistical mechanics as the thermal equilibrium distribution of wave functions \cite{GLTZ06}, and wave function collapse, either in the form of a collapse postulate in the quantum formalism or in collapse theories \cite{BG03,GB25}, \cite[Sec.~3.3]{Tum22}, \cite{GRW86,Bell87,Dio87,Pea89,GPR90} (i.e., theories that provide foundations for the quantum formalism by postulating equations for how the collapse occurs).

For every density matrix $\rho$ on a Hilbert space $\Hilbert$, there is a unique measure $\GAP_\rho$ on the unit sphere
\be
\SSS(\Hilbert) := \{ \psi \in \Hilbert: \|\psi\|=1\}\,.
\ee
In a precise sense, it is the most spread-out probability distribution with density matrix $\rho$ \cite{JRW94}. These measures are also known under the name ``Scrooge measure'' \cite{JRW94,Mar24,MS25,Pre26}, and their properties and applications have been studied extensively \cite{TZ05,Rei08,GLMTZ16,Tum20,TTV25,Mar24,Vog25,MS25,Pre26,ITV26}. For their definition, see Section~\ref{sec:main}. For $\rho$ proportional to a projection $P$ (i.e., for $\rho=P/\tr P$), $\GAP_\rho$ is in fact the uniform distribution over $\SSS(\operatorname{range} P)$.

Our result applies to various kinds of collapses:
\begin{itemize}
\item First, to the collapse associated with the ideal quantum measurement of a self-adjoint observable with discrete spectrum,
\be\label{Adef}
A=\sum_{\alpha\in\mathrm{spectrum}(A)} \alpha P_\alpha\,,
\ee
where $P_\alpha$ is the projection to the eigenspace with eigenvalue $\alpha$; the collapse then replaces the wave function $\Psi$ of a system by the collapsed wave function
\be\label{Acollapse}
\Psi' = \frac{P_\alpha \Psi}{\|P_\alpha \Psi\|}
\ee
if $\alpha$ was the observed value.

\item Second, to more general collapses of the form
\be\label{generalmeasure}
\Psi' = \frac{L_z \Psi}{\|L_z \Psi\|} \,
\ee
where $L_z$ can be any operator (not necessarily a projection) describing how to collapse $\Psi$ upon observing the value $z$, subject to the condition
\be
\sum_z L_z^\dagger L_z = I
\ee
needed to ensure that the associated formula for the probability distribution of the outcome $z$,
\be\label{generalmeasureBorn}
\PPP(z) = \| L_z \Psi \|^2
\ee
does define a probability distribution.

\item Third, to collapse theories \cite{BG03,GB25} such as GRW \cite{GRW86,Bell87} \cite[Sec.~3.3]{Tum22}, CSL \cite{Pea89,GPR90,BG03,GB25}, and the Di\'osi-Penrose model \cite{Dio87}. In these cases we condition, not on the observed outcome, but on the times and locations of the GRW collapses, respectively on the history of the noise in CSL.
\end{itemize}

Briefly put, our result asserts that if a wave function $\Psi$ is $\GAP_\rho$ distributed and a collapse occurs, then the distribution of the collapsed wave function $\Psi'$, conditionally on \emph{which} collapse occurred, is $\GAP_{\rho'}$ with the appropriate $\rho'$. 

There is another, well-known way in which GAP measures are linked to quantum measurements (see Section~\ref{sec:cond}): if a quantum system (let us call it the $a$-system) is entangled with another system $b$ and a quantum measurement relative to a purely random orthonormal basis $B$ is carried out on $b$, then the collapsed state of $a$ (known as the \emph{conditional wave function}) has distribution $\GAP_\rho$ when averaged over $B$, with $\rho$ the reduced density matrix of $a$. In fact, this relation will be used for an alternative (less direct but perhaps more intuitive) proof of our main result in the case $\dim\Hilbert<\infty$ in Section~\ref{sec:cond}.

A rather different fact linking GAP measures and collapse was recently observed in \cite{BLV26}: the authors found that in an infinite sequence of quantum measurements of randomly chosen observables, the distribution of the collapsed wave function converges in certain cases to a GAP measure.

The remainder of this paper is organized as follows. In Section~\ref{sec:main}, we recall the definition of the GAP measure and formulate and prove our result as a mathematical theorem. In Section~\ref{sec:discussion}, we describe some examples of how collapse theories are covered by our theorem and provide further discussion. In Section~\ref{sec:cond}, we describe an alternative proof of our theorem based on the concept of conditional wave function.

\section{Main Result}
\label{sec:main}

GAP measures can be defined as follows. Let $\rho$ be a density matrix (i.e., a positive operator with trace 1) and $\G_\rho$ the complex-Gaussian measure on $\Hilbert$ with mean 0 and covariance operator $\rho$ \cite{GLTZ06,Tum20}.\footnote{Here, a complex-Gaussian measure in $\CCC$ with mean 0 is understood as circularly symmetric (i.e., real and imaginary part of the random variable are identically distributed); a complex-Gaussian measure in a Hilbert space is a Gaussian measure such that each component of the random vector relative to any orthonormal basis is complex-Gaussian in $\CCC$. In fact, there is a unique complex-Gaussian measure for any given mean $\mu\in\Hilbert$ and covariance $\rho$ that is a positive trace-class operator on $\Hilbert$ \cite{Pro56} (see also \cite{Tum20} for discussion).} Define $\GA_\rho$ as the measure on $\Hilbert$ with density function $\|\cdot\|^2$ relative to $\G_\rho$,
\be
\GA_\rho(d\psi) = \|\psi\|^2 \: \G_\rho(d\psi)\,.
\ee
Using $\tr\rho=1$, one finds that $\GA_\rho$ is a probability measure \cite{GLTZ06}. Let $\Phi$ be a random vector with distribution $\GA_\rho$; then $\GAP_\rho$ is defined as the distribution of $\Psi:= \Phi/\|\Phi\|$. Clearly, $\GAP_\rho$ is concentrated on the unit sphere $\SSS(\Hilbert)$. In fact, the name ``GAP'' stands for ``Gaussian adjusted projected,'' where the adjustment is the multiplication by $\|\cdot\|^2$ and the projection maps to the unit sphere. For further discussion of GAP measures and how they arise in thermal equilibrium, see \cite{GLTZ06,GLMTZ16,TTV25}.

We can now formulate our main result as a theorem. We use the ``$\sim$'' notation as in ``$X\sim \mu$'' for expressing that the random variable $X$ has probability distribution $\mu$.

\begin{thm}\label{thm:1}
Let $\Hilbert$ be a Hilbert space and $\Psi$ a random point on $\SSS(\Hilbert)$ (the ``initial wave function''). Suppose for each point $x$ in the measure space $(\sX,\sF,\mu)$, we are given a bounded operator $L(x)$ (``collapse operator''), depending $\sF$-measurably on $x$, such that
\be\label{normalization}
\int_{\sX} \mu(dx) \, L^\dagger(x) \, L(x) = I \,.
\ee
Suppose the random point $X$ gets chosen with probability distribution
\be\label{Born}
\|L(x) \Psi\|^2 \, \mu(dx)
\ee
(``Born distribution''), and define
\be\label{Psi'def}
\Psi'=\frac{L(X) \Psi}{\|L(X) \Psi \|}
\ee
(the ``final wave function''). If $\Psi \sim \GAP_\rho$, then the conditional distribution of $\Psi'$, given $X$, has probability 1 to be $\GAP_{\rho'(X)}$ with 
\be\label{rho'(x)}
\rho'(x)=\frac{L(x)\rho L^\dagger(x)}{\tr[L(x)\rho L^\dagger(x)]} \,.
\ee
\end{thm}

\bigskip

\begin{proof}
As a preparation, we collect some formulas for later use: First, we re-derive the well known fact \cite{GLTZ06,GLMTZ16} that $\GAP_\rho$ has covariance $\rho$, i.e., (since it has mean 0)
\be\label{rhoGAPrho}
\int_{\SSS(\Hilbert)} \GAP_\rho(d\psi) \: |\psi\rangle \langle\psi| =\rho\,.
\ee
Indeed, by definition
\be
\int_{\SSS(\Hilbert)} \GAP_\rho(d\psi) \: |\psi\rangle \langle\psi| 
=\int_{\Hilbert} \GA_\rho(d\phi) \: \frac{|\phi\rangle \langle\phi|}{\|\phi\|^2} 
= \int_{\Hilbert} \G_\rho(d\phi) \: |\phi\rangle \langle\phi|
=\rho\,,
\ee
as $\G_\rho$ has mean 0 and covariance $\rho$.

Next, by construction,
\begin{align}
\PPP(X\in dx|\Psi) &= \|L(x) \Psi\|^2 \, \mu(dx)\\
\PPP(X\in dx, \Psi \in d\psi) &= \|L(x) \psi\|^2 \, \mu(dx) \, \GAP_\rho(d\psi)\\
\PPP(X\in dx) &= \mu(dx) \int_{\SSS(\Hilbert)} \GAP_\rho(d\psi) \: \|L(x) \psi\|^2\\
&= \mu(dx) \int_{\SSS(\Hilbert)} \GAP_\rho(d\psi) \: \tr\Bigl(|\psi\rangle \langle\psi| L^\dagger(x) L(x) \Bigr)\\
&= \mu(dx) \tr\Biggl[ \Bigl(\int_{\SSS(\Hilbert)} \GAP_\rho(d\psi) |\psi\rangle \langle\psi| \Bigr) L^\dagger(x) L(x) \Biggr]\\
&\stackrel{\eqref{rhoGAPrho}}{=} \mu(dx) \tr\bigl[ \rho L^\dagger(x) L(x) \bigr]\,.
\end{align}
The quantity we are ultimately interested in is the conditional distribution of $\Psi'$, given $X$; in formulas, this distribution is, for any $B\subset \SSS(\Hilbert)$,
\begin{align}
\PPP(\Psi' \in B |X\in dx) 
&= \PPP(\Psi' \in B , X\in dx)/\PPP(X \in dx)\\
&= \PPP\biggl( \frac{L(X) \Psi}{\|L(X) \Psi \|} \in B ,X \in dx\biggr)/\PPP(X\in dx)\,.
\end{align}
Write $\Psi=\Phi/\|\Phi\|$ with $\Phi \sim \GA_\rho$. Then
\begin{align}
&\PPP\biggl(\frac{L(x) \Psi}{\|L(x) \Psi \|} \in B ,X \in dx\biggr)\nonumber\\
&\qquad= \PPP\biggl(\frac{L(x) \Phi}{\|L(x) \Phi \|} \in B ,X \in dx\biggr)\\
&\qquad= \int_{\phi\in\Hilbert} \PPP(\Phi \in d\phi) \: \PPP\biggl(\frac{L(x) \Phi}{\|L(x) \Phi \|} \in B ,X \in dx \bigg| \Phi=\phi \biggr)\\
&\qquad= \int_{\phi\in\Hilbert} \GA_\rho(d\phi) \: 1_B\biggl(\frac{L(x) \phi}{\|L(x) \phi \|}\biggr) \: \PPP(X \in dx|\Phi=\phi)\\
&\qquad= \int_{\phi\in\Hilbert} \GA_\rho(d\phi) \: 1_B\biggl(\frac{L(x) \phi}{\|L(x) \phi \|} \biggr) \: \frac{\|L(x) \phi\|^2}{\|\phi\|^2} \: \mu(dx)\\
\intertext{[using $\GA_\rho(d\phi) = \|\phi\|^2 \: \G_\rho(d\phi)$]}
&\qquad= \mu(dx) \int_{\phi\in\Hilbert} \G_\rho(d\phi) \: 1_B\biggl(\frac{L(x) \phi}{\|L(x) \phi \|} \biggr) \: \|L(x) \phi\|^2\,.\label{pf5last}
\end{align}
Now substitute $\xi = L(x) \phi$. Note that if $\phi$ is Gaussian with mean 0, then for any operator $L$, $L \phi$ is also Gaussian with mean 0, and its covariance matrix is $\EEE [L|\phi\rangle \langle \phi| L^\dagger]= L\EEE [|\phi\rangle \langle \phi|] L^\dagger=L\rho L^\dagger$ with $\rho$ the covariance of $\phi$. Thus, $\xi \sim \G_{L(x)\rho L^\dagger(x)}$, and
\be\label{pf6}
\eqref{pf5last}= \mu(dx) \int_{\xi\in\Hilbert} \G_{L(x) \rho L^\dagger(x)}(d\xi) \: 1_B\biggl(\frac{\xi}{\|\xi \|} \biggr)  \:\|\xi\|^2 \,.
\ee
Now substitute $\chi = \xi/\sqrt{\tr[ \rho L^\dagger(x) L(x) ]}$, that is, we just rescale $\xi$ by a fixed constant factor; then $\chi$ is also Gaussian, with mean 0 and covariance matrix
\be
L(x) \rho L^\dagger(x)/\tr[ \rho L^\dagger(x) L(x) ]=\rho'(x) \,.
\ee
Thus,
\begin{align}
\eqref{pf6}
&= \mu(dx) \int_{\chi\in\Hilbert} \G_{\rho'(x)}(d\chi) \: 1_B(\chi/\|\chi \|) \: \|\chi\|^2 \: \tr[ \rho L^\dagger(x) L(x) ]\\
\intertext{[using the definition of GA]}
&= \mu(dx) \tr[ \rho L^\dagger(x) L(x) ] \int_{\chi\in\Hilbert} \GA_{\rho'(x)}(d\chi) \:  1_B(\chi/\|\chi \|)\\
\intertext{[substituting $\varphi=\chi/\|\chi\|$ and using the definition of GAP]}
&= \mu(dx) \tr[ \rho L^\dagger(x) L(x) ] \int_{\varphi\in\SSS(\Hilbert)} \GAP_{\rho'(x)}(d\varphi) \: 1_B(\varphi)\\
&= \mu(dx) \tr[ \rho L^\dagger(x) L(x) ] \: \GAP_{\rho'(x)}(B)\,.
\end{align}
Thus,
\begin{align}
\PPP(\Psi' \in B |X\in dx)
&= \frac{ \mu(dx) \: \tr[ \rho L^\dagger(x) L(x) ] \: \GAP_{\rho'(x)}(B ) }{ \mu(dx) \: \tr[ \rho L^\dagger(x) L(x) ] }\\
&= \GAP_{\rho'(x)}(B )\,,
\end{align}
which is what the theorem claimed.
\end{proof}

\section{Examples}
\label{sec:discussion}

We describe how the theorem covers the examples mentioned before. 
\begin{itemize}
\item First, for an {\bf ideal quantum measurement} of \eqref{Adef}, $\sX= \mathrm{spectrum}(A)$ discrete, the $\sigma$-algebra $\sF$ contains all subsets, and $\mu$ is the counting measure, so that $\int_\sX \mu(dx)$ means the same as $\sum_\alpha$. The operator $L(x)$ is given by $P_\alpha$, and \eqref{normalization} is satisfied because
\be
\sum_\alpha P_\alpha = I\,.
\ee
The random value $X$ is the measurement outcome, the distribution \eqref{Born} is indeed the Born distribution giving weight $\langle \Psi |P_\alpha|\Psi\rangle = \|P_\alpha \Psi\|^2$ to each eigenvalue $\alpha$ of $A$, and the collapse formula \eqref{Psi'def} reduces to \eqref{Acollapse}.

\item For the {\bf more general measurements} of the form \eqref{generalmeasure}, $\sX$ is still the discrete set of possible outcomes, $X$ the actual outcome, $L(x)$ means the same as $L_z$, and \eqref{Born} still reduces to the Born distribution \eqref{generalmeasureBorn}, which corresponds to a POVM with the positive operator $L^\dagger_z L_z$ associated to the possible outcome $z$.

\item For {\bf GRW theory}, 
we consider as $\Psi'$ the wave function $\Psi_\tau\in L^2(\RRR^{3N})$ obtained through the stochastic GRW evolution at time $t=\tau>0$ and as $\Psi$ the wave function at $t=0$; $x$ is a particular history of collapses, given by specifying the number $n$ of collapses that occur during the time interval $[0,\tau]$, the times $0\leq t_1<\ldots<t_n\leq \tau$ at which the collapses occur, the locations\footnote{Here and in the following, bold-face symbols $\vx$ denote 3-dimensional vectors, whereas the plain italic symbol $x$ denotes a collapse history.} $\vx_1,\ldots,\vx_n \in \RRR^3$ at which the collapses are centered, and the labels $i_1,\ldots,i_n\in\{1,\ldots,N\}$ of the particles subject to each collapse. Correspondingly,
\be
\sX=\bigcup_{n=0}^\infty \sX_n
\ee
with
\be
\sX_n = \Bigl\{(t_1,\vx_1,i_1,\ldots,t_n,\vx_n,i_n) \in \bigl([0,\tau]\times \RRR^3 \times \{1...N\}\bigr)^n: t_1<\ldots <t_n  \Bigr\}
\ee
and measure
\begin{align}
\mu(B)&=\sum_{n=0}^\infty \mu_n(B\cap \sX_n)\,,\\
\mu_n(B_n)&=  \int_0^\tau dt_1\int_{t_1}^\tau dt_2\cdots \int_{t_{n-1}}^\tau \!\!\! dt_n\: (N\lambda)^n \, e^{-N\lambda \tau}  \:\times \nonumber \\
&\qquad \times \: \int_{\RRR^3}d^3\vx_1\cdots \int_{\RRR^3} d^3\vx_n \sum_{i_1...i_n=1}^N \frac{1}{N^n} 1_{B_n}(t_1,\vx_1,i_1,\ldots,t_n,\vx_n,i_n)\label{mundef}
\end{align}
with $\lambda$ the collapse rate per particle (one of the constants of the GRW theory). The collapse operator for $x=(t_1,\vx_1,i_1,\ldots,t_n,\vx_n,i_n)$ is
\be\label{GRWLxdef}
L(x) = U_{t_n}^\tau C_{i_n}(\vx_n) U^{t_n}_{t_{n-1}} \cdots C_{i_2}(\vx_2)U_{t_1}^{t_2}C_{i_1}(\vx_1)U_0^{t_1}
\ee
with
\be
U_t^{t'}= e^{-iH(t'-t)/\hbar}
\ee
the unitary part of the time evolution from $t$ to $t'$ and $C_i(\vx)$ the multiplication operator by (the square root of) a Gaussian in the $i$-th variable centered at $\vx$,
\be
C_i(\vx) \, \Psi(\vq_1,\ldots,\vq_N) = (2\pi\sigma^2)^{-3/4} \exp\Bigl(-\frac{(\vq_i-\vx)^2}{4\sigma^2} \Bigr) \: \Psi(\vq_1,\ldots,\vq_N)\,,
\ee 
where $\sigma$ is the collapse width (another constant of GRW theory).\footnote{Alternatively, we would obtain an equivalent theory if we remove the factor $(N\lambda)^n \, e^{-N\lambda \tau}$ from \eqref{mundef} and include its square root in the definition \eqref{GRWLxdef} of $L(x)$.}

Then \eqref{Born} is exactly the joint distribution of all collapse events (e.g., \cite[(5.27)]{Tum22}) according to GRW theory, and $\Psi'$ as in \eqref{Psi'def} is indeed equal to $\Psi_\tau$.

\item For {\bf CSL theory}, $\Psi$ is again the wave function at $t=0$ and $\Psi'$ the one at $t=\tau$; $x$ is the ``noise field'' $\xi(\vx,t)$ for all $\vx\in\RRR^3, t\in[0,\tau]$, and $\mu$ the distribution of white noise in $\RRR^3\times [0,\tau]$ with mean 0 and correlation
\be
\EEE \bigl[ \xi(\vx,t) \, \xi(\vx',t') \bigr] = \gamma\, \delta^3(\vx-\vx') \, \delta(t-t') \,,
\ee
$\gamma>0$ a constant; $L(x) \psi$ is the solution at time $\tau$ of the (Stratonovich-type) equation \cite[(8.6)]{BG03}
\be\label{Stratonovich}
\frac{d}{dt}\psi(t) = \Biggl[ -\frac{i}{\hbar}H + \int_{\RRR^3}d^3\vx \, N(\vx) \, \xi(\vx,t) - \gamma \int_{\RRR^3} d^3\vx \, N^2(\vx) \Biggr] \psi(t)
\ee
starting from $\psi(0)=\psi$ with $N(\vx)$ the smeared-out particle number density operator,
\be
N(\vx) \, \psi(\vq_1,\ldots,\vq_N) = \sum_{i=1}^N (2\pi\sigma^2)^{-3/2}\exp\Bigl(-\frac{(\vq_i-\vx)^2}{2\sigma^2} \Bigr) \, \psi(\vq_1,\ldots,\vq_N)\,.
\ee
In other words, $L(x)$ is the time-ordered exponential of the square bracket in \eqref{Stratonovich}; $X$ is the actual realization of the noise field; its distribution \eqref{Born} is traditionally called the ``cooked'' probability, while $\mu$ is called the ``raw'' probability \cite[(7.40)]{BG03}; the known fact that the ``cooked'' distribution is always normalized proves \eqref{normalization}; $\Psi'$ as given by \eqref{Psi'def} then agrees with what is called the ``physical state vector'' \cite[(7.42-3)]{BG03}.
\end{itemize}

\begin{rem}
If we do not want to condition on $x$ then, in order to obtain the (unconditional) distribution of $\Psi'$, we need to average over $x$:
\begin{align}
\PPP(\Psi'\in B) 
&= \int_{x\in\sX} \PPP(X\in dx) \: \PPP(\Psi'\in B|X=x)\\
&=\int_{x\in\sX} \mu(dx) \: \tr[\rho L^\dagger(x) L(x)] \: \GAP_{\rho'(x)}(B)\,.\label{uncond}
\end{align}
In particular, the distribution of $\Psi_\tau$ in GRW theory is a mixture of GAP measures.
\end{rem}

\begin{rem}
Since for any probability distribution $\pi$ on $\SSS(\Hilbert)$, the associated density matrix is
\be
\rho_\pi = \int_{\psi \in\SSS(\Hilbert)}\pi(d\psi) \: |\psi\rangle\langle\psi|\,,
\ee
and since this formula is linear in $\pi$, it follows that for any mixture of distributions over the unit sphere, the density matrix is the corresponding mixture of the density matrices of the contributing distributions. Applying this to \eqref{uncond}, we obtain that
\begin{align}
\rho_{\PPP(\Psi' \in \,\cdot)} 
&= \int_{x\in\sX} \mu(dx) \: \tr[\rho L^\dagger(x) L(x)] \: \rho'(x)\\
&\stackrel{\eqref{rho'(x)}}{=} \int_{x\in\sX} \mu(dx) \: L(x)\rho L^\dagger(x)\,,
\end{align}
which is the post-collapse density matrix. In particular for GRW theory, this is the density matrix of $\Psi_\tau$, i.e., the solution of the GRW master equation \cite[(6.8)]{BG03}.
\end{rem}

\section{Alternative Proof}
\label{sec:cond}

In this section, we present a different strategy for proving Theorem~\ref{thm:1} in the case $\dim\Hilbert<\infty$; it is less direct but maybe more intuitive, and based on conditional wave functions. The conditional wave function $\Psi_\mathrm{cond}$ is a wave function attributed to a system (let us call it the $a$-system) that is entangled with another system (let us call it the $b$-system); it was first introduced in Bohmian mechanics \cite{LD90,DGZ92} as
\be\label{PsiconddefBohm}
\Psi_\mathrm{cond}(q_a) = \mathcal{N} \, \Phi(q_a,Q_b) \,,
\ee
where $\Phi$ is the wave function of the full system $a\cup b$, $q_a$ the configuration variable of the $a$-system, $Q_b$ the Bohmian configuration of the $b$-system, and $\mathcal{N}$ the normalizing constant. More generally, for any $\Phi\in\SSS(\Hilbert_a \otimes \Hilbert_b)$ and any orthonormal basis $B=\{|b_j\rangle\}_j$ of $\Hilbert_b$, one defines
\be\label{Psiconddef}
\Psi_\mathrm{cond} = \mathcal{N} \, \langle b_J|\Phi\rangle_b \in \SSS(\Hilbert_a)\,,
\ee
where $\langle \cdot|\cdot \rangle_b$ is the partial inner product over $b$ and $|b_J\rangle$ is a random basis vector chosen with Born distribution
\be\label{JBorn}
\PPP(J=j) = \bigl\|\langle b_j|\Phi\rangle_b \bigr\|_a^2\,.
\ee
One can think of $\Psi_\mathrm{cond}$ as the collapsed wave function of $a$ after a quantum measurement relative to $B$ has been carried out on $b$.

It is known \cite[Lemma 1]{GLMTZ16} that, for $\dim\Hilbert_b<\infty$, $\GAP_\rho$ is the distribution of the conditional wave function, averaged over all possible $B$'s, when $\rho= \tr_b |\Phi\rangle\langle \Phi|$: 
\be\label{EcondGAP}
\EEE_B \mathscr{L}(\Psi_\mathrm{cond})= \GAP_\rho\,,
\ee
where $\mathscr{L}(Y)$ means the distribution (``law'') of the random variable $Y$. Here, the measure for $B$ is the uniform distribution $u_{ONB}$ over all orthonormal bases of $\Hilbert_b$, which is closely related to the Haar measure over the unitary group of $\Hilbert_b$ and exists whenever $\dim\Hilbert_b<\infty$.

So consider the case in which $\Hilbert_a$ is the $\Hilbert$ of Theorem~\ref{thm:1}, $\dim \Hilbert_a \leq \dim \Hilbert_b<\infty$. For given $\rho$, choose any $\Phi\in\SSS(\Hilbert_a \otimes \Hilbert_b)$ with $\tr_b |\Phi\rangle \langle\Phi|=\rho$; e.g., $\Phi= \sum_i \sqrt{p_i} \, \varphi_{ai} \otimes \varphi_{bi}$ with ONBs $\{\varphi_{xi}\}_i$ of $\Hilbert_x$ ($x=a,b$) such that $\rho= \sum_i p_i |\varphi_{ai}\rangle \langle \varphi_{ai}|$. Consider choosing a random $B\sim u_{ONB}$, then carrying out a quantum measurement of $B$ on $b$ with outcome $J\sim\eqref{JBorn}$ and collapsed state $\Phi' = \Psi \otimes |b_J\rangle$ with $\Psi = \Psi_\mathrm{cond}= \mathcal{N} \, \langle b_J|\Phi\rangle_b$, and then carrying out a generalized quantum measurement with collapse operators $L$ on the $a$-system, with outcome $X\sim\eqref{Born}$ and collapsed state $\Phi'' = \Psi'' \otimes |b_J\rangle$. We know that $\Psi \sim \GAP_\rho$ and want to compute the distribution of $\Psi''=\mathcal{N}'' L(X) \Psi=\langle b_J|\Phi'' \rangle_b$, conditionally on $X$.

From the fact that the two measurements act on separate systems, one on $b$ and one on $a$, it is clear that the distribution of the pair of outcomes $(J,X)$ (as well as that of $\Phi''$) are the same as if the two measurements had been carried out in the reverse order. This can also be easily checked directly, as 
\be\label{JXjoint}
\PPP(J=j,X\in dx) = \bigl\|(L(x)\otimes |b_j\rangle\langle b_j|)\Phi \bigr\|^2 \mu(dx)
\ee
for any order of the measurements, and $\Phi''$ is fixed according to
\be\label{PhiJX}
\Phi'' = \mathcal{N}'' (L(X)\otimes |b_J\rangle\langle b_J|)\Phi
\ee
once $B$, $J$, and $X$ are chosen. Now consider the reverse order: first, a quantum measurement of $L$ on $a$ collapses $\Phi$ to 
\be
\widetilde{\Phi} = \widetilde{\mathcal{N}} \, (L(X)\otimes I_b)\Phi
\ee
with
\be
\PPP(X\in dx) = \| (L(x) \otimes I_b) \Phi\|^2 \mu(dx)\,,\quad \widetilde{\mathcal{N}}=\| (L(x) \otimes I_b) \Phi\|^{-1};
\ee
second, a subsequent $B$-measurement leads to (as we saw) $\Phi''$ and a conditional wave function $\widetilde{\Psi}=\langle b_J|\Phi''\rangle_b$ such that, by \eqref{EcondGAP}, 
\be
\EEE_B\mathscr{L}(\widetilde{\Psi})= \GAP_{\widetilde{\rho}(X)}
\ee
with
\begin{align}
\widetilde{\rho}(X)
&= \tr_b |\widetilde{\Phi}\rangle \langle \widetilde{\Phi}| \\
&= |\widetilde{\mathcal{N}}|^2 \, \tr_b \Bigl[(L(X)\otimes I_b)|\Phi \rangle \langle \Phi| (L^\dagger(X)\otimes I_b) \Bigr] \\
&= \| (L(X) \otimes I_b) \Phi\|^{-2} \, L(X)\tr_b \bigl[|\Phi \rangle \langle \Phi| \bigr] L^\dagger(X)\\
&= \frac{L(X) \rho L^\dagger(X)}{\tr [L(X) \rho L^\dagger(X)]}=\rho'(X)\,.
\end{align}
As we saw above, $\mathscr{L}(\Psi'',B,J,X)=\mathscr{L}(\langle b_J|\Phi''\rangle_b,B,J,X)=\mathscr{L}(\widetilde{\Psi},B,J,X)$, so, averaging over $B$ and $J$,
\be
\mathscr{L}(\Psi''|X)=\GAP_{\rho'(X)}
\ee
with probability 1, which is what we wanted to prove.

\end{document}